\UseRawInputEncoding
\documentclass[%
 reprint,
 showkeys,
 showpacs,
 superscriptaddress,
 amsmath,amssymb,
 aps,
pra,
]{revtex4-2}

\usepackage{url}
\usepackage{amsfonts}
\usepackage{mathrsfs}
\usepackage{amssymb}
\usepackage{amsthm}
\usepackage{physics}
\usepackage{graphicx}
\usepackage{dcolumn}
\usepackage{bm}
\usepackage{hyperref}
\usepackage{comment}
\usepackage{caption}
\usepackage{subcaption}
\usepackage{mathtools}
\usepackage{multirow}
\hypersetup{
    colorlinks,
    citecolor=blue,
    filecolor=blue,
    linkcolor=blue,
    urlcolor=blue
}

\newtheorem{theorem}{Theorem}[section]

\newtheorem{lemma}{Lemma}[section]

\newtheorem{definition}{Definition}[section]

\begin{document}

\DeclareRobustCommand{\pt}{\ensuremath{\mathcal{PT}}}

\preprint{APS/123-QED}

\title{Interaction-Endowed \pt-Symmetry and its Effects on Decoherence, Einselection, and Non-Markovianity in a Central Spin Model}

\author{Leoj Phoebe M. Esquierdo}
\affiliation{Department of Physics, School of Foundational Studies and Education, Mapúa University, Intramuros, Manila, Philippines}
\author{Lemuel John F. Sese}
\affiliation{Department of Physics, Pohang University of Science and Technology, Pohang, Korea}
\author{Rayda P. Gammag}
\email{rpgammag@mapua.edu.ph}
\affiliation{Department of Physics, School of Foundational Studies and Education, Mapúa University, Intramuros, Manila, Philippines}

\date{\today}

\begin{abstract}
We have introduced \pt-symmetry to a central spin model by adding a \pt-symmetric interaction with a tunable hermiticity parameter $\gamma$. Using the pseudo-Hermitian formalism, we applied a Dyson map to transform the non-Hermitian Hamiltonian to its Hermitian representation. We have found that the decoherence slows down as $\gamma$ increases, eventually ceasing at $\gamma=1$. We define a pseudo-Hermitian observable that commutes with the metric operator to add as the self-Hamiltonian. Einselection gradually forced the system to select the eigenstates of the self-Hamiltonian as $\gamma\rightarrow1$. The steady-state purity of the central spin states in the strong environment regime exhibits a paradoxical decrease as $\gamma$ increases. However, a turning point (minimum) corresponding to a maximum information dissipation to the spin bath is found. Finally, the Breuer-Laine-Piilo (BLP) measure, which is used to quantify the non-Markovianity of a quantum system, was evaluated for a finite time. The BLP measure in the strong environment regime exhibited similar turning point behavior, which means that the information backflow reaches a saturation point before declining. This decline signifies the point where \pt-symmetry starts shielding the central spin from the environment.
\end{abstract}

\maketitle


\section{\label{Introduction}Introduction}

In open quantum systems, a specific system of interest interacts with its environment, often called a bath. This generates entanglement between system and bath states, which inevitably destroys the coherences of the system states \cite{Joos1985TheEO}. However, depending on the nature of the interaction, certain observables get 'measured' by the bath \cite{Zurek_1981}. This monitoring paradoxically preserves a set of selected states and entangles it with specific bath states. This process is known as environmentally-induced superselection, or einselection, which results in a set of states called the pointer states \cite{Zurek_1982}. This concrete theory is known as the decoherence theory, which was first formulated to explain the emergence of classical behavior from quantum systems.

A well-known model for decoherence theory is a spin-1/2 system interacting with many bath spins through a dephasing interaction, also known as the central spin model \cite{Zurek_1982}. This was first used to prove einselection, where it was shown that the bath actively selects the eigenstates of the interaction Hamiltonian as pointer states. In addition, this was expanded to include a self-Hamiltonian, which describes the internal dynamics of the central spin. In this extended model, the self-Hamiltonian and the interaction Hamiltonian compete with each other to determine the pointer states. The rate of decoherence for this model was found to be Gaussian in time, given that the coupling energies are independent and have finite variance \cite{Cucchietti}.

On the other hand, parity-time-symmetric (\pt-symmetric) extensions have also been applied to decoherence models, particularly in spin-boson models, where a spin-1/2 system interacts with a thermal bath, modeled as bosonic modes \cite{Gardas,Zhang_Chen,Dey}. $\mathcal{P}$ and $\mathcal{T}$ operators are the parity-inversion and time-reversal operators, respectively. The parity-inversion operator flips the sign of the spatial coordinate, while the time-reversal operator reverses the direction of time. \pt-symmetric systems are invariant under the combined action of these operators, which are typically modeled as systems with balanced gain and loss \cite{Musslimani_Makris_El-Ganainy_Christodoulides_2008, Makris_El-Ganainy_Christodoulides_Musslimani_2008}. Non-Hermitian Hamiltonians that possess this symmetry have real eigenenergies.

Through the pseudo-Hermitian formalism, a Dyson transform can be applied to these Hamiltonians to map them to an equivalent Hermitian representation. However, this mapping is only defined as long as an antilinear symmetry is present in the non-Hermitian Hamiltonian, which is \pt \ in this case. Nevertheless, this formalism is enough to elucidate the interesting phenomena that \pt-symmetry has to offer. \pt-symmetry has been applied by Gardas, Deffner, and Saxena \cite{Gardas} to the spin-boson model through an identical \pt-symmetric self-Hamiltonian and interaction Hamiltonian acting on the central spin. The central spin exhibited a critical slowing down in decoherence as the system approaches the boundary between the unbroken and broken \pt-symmetry regimes, where decoherence freezes out completely. Another study carried out by Dey, Raj, and Goyal \cite{Dey} explored the case where the bath Hamiltonian is endowed with \pt-symmetry. This enhanced the central spin coherence time, even with weak couplings. These studies show that in the presence of an unbroken \pt-symmetry, coherences are preserved for much longer timescales. Such results imply that \pt-symmetry can offer a novel strategy for controlling and mitigating decoherence. Unfortunately, \pt-symmetric extensions of central spin model are sparse.

The biggest difference between spin-boson model and central spin model is that the latter naturally exhibits non-Markovian characteristics \cite{Prokofev2000}. To be specific, the central spin model is marked by a dependence of its current state to its past state. When coupled to a spin bath, the information of the past configurations of the central spin that spreads into its environment can backflow. This can affect its current state, which often leads to recoherence and longer decoherence times compared to when it interacts with a thermal bath \cite{Breuer_Laine_Piilo_2009,Jing}. In addition, a recent study has shown that spin-boson model endowed with \pt-symmetry shared a similar behavior \cite{Zhang_Chen}. However, this study was not able to show whether \pt-symmetry can induce non-Markovian features.

In this paper, we propose an ansatz interaction Hamiltonian, $\hat{H}_{\mathcal{SE}}$, endowed with \pt-symmetry. We then extend $\hat{H}_{\mathcal{SE}}$ to include a transverse self-Hamiltonian. Its motivation, along with further details on the construction of this Hamiltonian, is discussed in Section \ref{Heading2}. We show that, upon adding this self-Hamiltonian, the \pt-transition point becomes distinct from the exceptional point: the former marks the boundary between the broken and unbroken \pt-symmetry regimes, while the latter is where the eigenvalues and eigenstates coalesce. We restrict our analysis to the unbroken \pt-symmetry regime, where the eigenenergies are real, which justifies the use of conventional quantum mechanics via the Hermitian representation. Using this formalism, in Section \ref{Heading3} we explore whether effects such as recoherence and suppression of decoherence arise in the central spin model when endowed with a \pt-symmetric interaction. In Section \ref{Heading4}, we also investigate the effect of unbroken \pt-symmetry on einselection, motivated by the fact that the interaction Hamiltonian in this model does not commute with the self-Hamiltonian yet preserves the \pt-symmetric structure. Finally, in Section \ref{Heading5} we examine the trend in non-Markovianity arising from information backflow between the spin bath and the central spin, to further elucidate the mechanisms driving the observed changes in decoherence and einselection in this model.

\section{\label{Heading2}\pt-symmetric Central Spin Model}

The Hamiltonian of interest describes the interaction of the system $\mathcal{S}$ with an environment $\mathcal{E}$. It can be written as
\begin{equation}
    \hat{H}_{\mathcal{SE}}=\hat{V}_{\mathcal{S}}\otimes \hat{V}_{\mathcal{E}}=\frac{1}{2}\sum_{k=1}^{N}g_{k}\left(i\gamma\hat{\sigma}_{0}^{y}+\hat{\sigma}_{0}^{z}\right)\otimes \hat{\sigma}_{k}^{z},
\end{equation}
where $g_{k}$ is the coupling constant between the central spin and the $k$th bath spin, and $\gamma$ is the Hermiticity parameter. This Hamiltonian is a \pt-symmetric extension of the dephasing interaction used in \cite{Cucchietti}, where the \pt-symmetry is incorporated through the addition of $i\gamma\hat{\sigma}_{0}^{y}$ to $\hat{V}_{S}$.

Any \pt-symmetric Hamiltonian can be mapped to an equivalent Hermitian representation through a Dyson map $T$, such that
\begin{equation}
    \hat{h}=T\hat{H}T^{-1},
    \label{push_forward}
\end{equation}
where $\hat{h}$ is the Hermitian representation of $\hat{H}$. This Dyson transformation can be derived from the metric operator $\eta$ specific to that Hamiltonian. This can be done by exploiting the algebra of the spin operators. The metric operator must be Hermitian, hence we let
\begin{equation}
    \eta=e^{-Q},
    \label{metric_operator}
\end{equation}
where $Q=\sum_{i}q_{i}\hat{\sigma}_{0}^{i}$, $\hat{\sigma}_{0}^{i}$ being the $i$th-Pauli spin operator acting on the central spin.

To identify $Q$, we recall the condition of self-adjointness in pseudo-Hermitian formalism,
\begin{equation}
    \hat{O^{\#}}=\hat{O},
    \label{self-adjointness}
\end{equation}
where $\hat{O}^{\#}$ is the pseudo-Hermitian adjoint of an arbitrary non-Hermitian operator defined as
\begin{equation}
    \hat{O}^{\#}=\eta \hat{O}^{\dagger}\eta^{-1}.
\end{equation}
This ensures that the Dyson map will transform the given operator into a Hermitian representation.

Now, we can apply this self-adjointness condition to derive the metric operator for $V_{\mathcal{S}}$,
\begin{equation}
    \hat{V_{\mathcal{S}}}=e^{Q}\hat{V_{\mathcal{S}}}^{\dagger}e^{-Q}.
    \label{adjoint}
\end{equation}
A reasonable ansatz for $Q$ is $\theta\hat{\sigma}_{0}^{x}$, where $\theta\in \mathbb{R}$ is an unknown coefficient, since $\hat{V}_{\mathcal{S}}$ involves both $\hat{\sigma}_{0}^{y}$ and $\hat{\sigma}_{0}^{z}$. Plugging this into equation (\ref{adjoint}),
\begin{equation}
    \hat{\sigma}_{0}^{z}+i\gamma\hat{\sigma}_{0}^{y}=e^{\theta\hat{\sigma}_{0}^{x}}\left(\hat{\sigma}_{0}^{z}-i\gamma\hat{\sigma}_{0}^{y}\right)e^{-\theta\hat{\sigma}_{0}^{x}}.
    \label{similarity_transformation}
\end{equation}
To get $\theta$, we can simplify equation (\ref{similarity_transformation}), hence
\begin{equation}
    \eta=e^{\tanh^{-1}{\left(\gamma\right)}\hat{\sigma}_{0}^{x}}.
    \label{eta}
\end{equation}
The Dyson map $T$ can then be easily solved 
\begin{equation}
    T=\sqrt{\eta}=e^{\frac{1}{2}\tanh^{-1}{\left(\gamma\right)}\hat{\sigma}_{0}^{x}}.
\end{equation}
Applying the Dyson map to $H_{\mathcal{SE}}$,
\begin{equation}
    \hat{h}_{\mathcal{SE}}=T\hat{H}_{\mathcal{SE}}T^{-1}=\frac{1}{2}\sum_{k=1}^{N}g_{k,\gamma}\hat{\sigma}_{0}^{z}\otimes\hat{\sigma}_{k}^{z}.
    \label{effective_coupling}
\end{equation}
Here, the \pt-symmetric Hamiltonian is mapped to a dephasing interaction with the rescaled couplings, $g_{k,\gamma}=\sqrt{1-\gamma^{2}}g_{k}$.

We want the system to include a self-Hamiltonian for the central spin. However, this is not as straightforward as in conventional quantum mechanics. Pseudo-Hermitian observables must be self-adjoint with respect to the pseudo-Hermitian adjoint. As such, these observables are defined as the inverse mapping of the observables of the Hermitian representation through the Dyson map defined in equation (\ref{push_forward}),
\begin{equation}
    O=T^{-1}oT.
\end{equation}
Here, $o$ and $O$ are observables of the Hermitian representation and the corresponding pseudo-Hermitian observable, respectively.

To be able to add a self-Hamiltonian to the dephasing interaction conveniently, the chosen self-Hamiltonian must simultaneously be an observable of the \pt-symmetric Hamiltonian and its Hermitian representation. By construction,
\begin{equation}
    [\eta,\hat{\sigma}_{0}^{x}]=0,
    \label{external_field}
\end{equation}
which means that $\hat{\sigma}_{0}^{x}$ is also a pseudo-Hermitian observable with respect to our metric operator, satisfying both conditions. Then, our extended Hamiltonian is
\begin{equation}
    \hat{h}=\hat{h}_{\mathcal{S}}+\hat{h}_\mathcal{SE}=\omega_{0}\hat{\sigma}_{0}^{x}+\frac{1}{2}\sqrt{1-\gamma^{2}}\hat{\sigma}_{0}^{z}\otimes\sum_{k=1}^{N}g_{k} \hat{\sigma}_{k}^{z}.
    \label{Htot}
\end{equation}
Note that the choice of self-Hamiltonian in equation (\ref{external_field}) is done to preserve the pseudo-Hermiticity of the original non-Hermitian Hamiltonian. Moreover, choosing a different observable as the self-Hamiltonian in the Hermitian representation will yield non-Hermitian operators as corresponding pseudo-Hermitian observables, which will complicate experimental designs intended to validate the theory.

The Hamiltonian in equation (\ref{Htot}) possesses a property that is unique in this construction, i.e., its \pt-transition point is distinct from its exceptional point. The \pt-transition point refers to point/s in the parameter space where \pt-symmetry is spontaneously broken, which is characterized by a divergent metric operator. For our metric operator in equation (\ref{eta}), this occurs at $\gamma=1$. On the other hand, the exceptional point refers to a point in the parameter space where some, if not all, eigenvectors and eigenvalues coalesce. At $\gamma=1$, the interaction Hamiltonian is defective. However, the total Hamiltonian is still diagonalizable due to the self-Hamiltonian. As such, the exceptional point can only appear when $\omega_{0}=0$. Thus, the exceptional point is at $(\gamma,\omega_{0})=(1,0)$.

Since the equivalence of the pseudo-Hermitian formalism with the Hermitian representation is proven in Appendix \ref{app:pseudo}, further calculations from here on will be done in the Hermitian representation.

\section{\label{Heading3}\pt-symmetry Effects on Decoherence}

The Hamiltonians in equations (\ref{effective_coupling}) and (\ref{Htot}) are strikingly similar to the Hamiltonian used in \cite{Cucchietti} to investigate the decoherence dynamics in the central spin model. In this section, we employ a similar methodology, outlined in Appendix \ref{Review_Cucchietti}, to explore the effect of \pt-symmetry on this dephasing interaction.

Upon transformation as shown in equation (\ref{effective_coupling}), the coupling energies are rescaled by a factor of $\sqrt{1-\gamma^{2}}$. To explicitly show this change, we consider a pure initial state of the bath spins in the computational basis
\begin{equation}
    \ket{\mathcal{E}(0)}=\sum_{n=0}^{2^{N}-1}c_{n}\ket{n}.
\end{equation}
The coupling energies are the eigenvalues of $V_{\mathcal{E}}$ when acting on the computational basis
\begin{align}
    V_{\mathcal{E}}\ket{n}=B_{n}\ket{n}=\left(\sum_{k=1}^{N}(-1)^{n_{k}}g_{k,\gamma}\right)\ket{n},
\end{align}
where $n_{k}$ is the $k$th-digit of the binary form of $n$.

In the unbroken \pt-symmetry regime, $\sqrt{1-\gamma^{2}}<1$, which means that the coupling energies are decreased. In the large bath limit ($N\gg1$), the coupling energy $B_{n}$ becomes a random variable $B$. The characteristic function $\Phi(B)$, which is the Fourier pair of the decoherence function, is the probability distribution of this random variable. The variance $s_{N}^{2}$ of the probability distribution absorbs the square of the scaling factor, i.e., $s_{N}^{2}\rightarrow(1-\gamma^{2})s_{N}^{2}$. This transforms the probability distribution to
\begin{equation}
    \Phi(B)=\frac{1}{\sqrt{2\pi s_{N}^{2}(1-\gamma^{2})}}\exp\left[-\frac{\left(B-\bar{B}_{N}\right)^{2}}{2s_{N}^{2}(1-\gamma^{2})}\right].
    \label{characteristic_function}
\end{equation}
Notice that as $\gamma\rightarrow1$,
\begin{equation}
    \lim_{\gamma\rightarrow1}\Phi(B)=\delta(B-\Bar{B}_{N}).
\end{equation}

Applying Fourier transform to the characteristic function, we can derive the decoherence factor,
\begin{equation}
    r(t)=\int e^{-iBt}\Phi(B)dB=e^{i\Bar{B}_{N}t}e^{-(1-\gamma^{2})s_{N}^{2}t^{2}/2},
\end{equation}
with its magnitude being
\begin{equation}
    |r(t)|=e^{-(1-\gamma^{2})s_{N}^{2}t^{2}/2}.
\end{equation}
The Gaussian profile of the decoherence is preserved, with a simple caveat, which is the additional scaling factor in the exponent of the decoherence factor.

We know that the states of the central spin will completely decohere when $|r(t)|=0$, but the $\left(1-\gamma^{2}\right)$ in the exponent will stretch the decoherence timescale such that $\tau_{D}=\sqrt{\frac{2}{(1-\gamma^{2})s_{N}^{2}}}$. This stretching of the decoherence timescale can be seen in Figure \ref{Decoherence_Factor_Plot}.

\begin{figure}[t]
    \centering
    \includegraphics[width=0.9\linewidth]{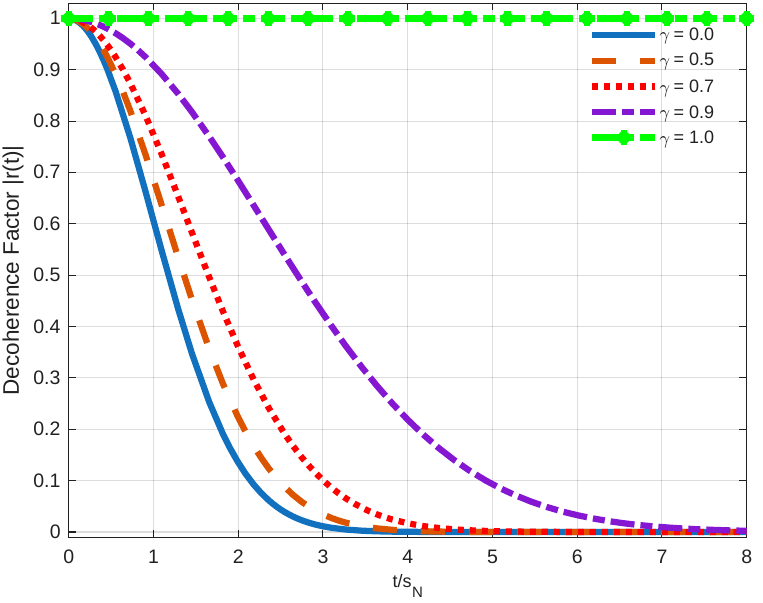}
    \caption{Decoherence factor for different values of $\gamma$. For $\gamma\in[0,1)$, the decoherence factor vanishes at $t\rightarrow\infty$ on larger timescales as $\gamma\rightarrow1$. At the \pt-transition point ($\gamma=1$), the decoherence factor becomes constant.}
    \label{Decoherence_Factor_Plot}
\end{figure}

Approaching the \pt-transition point, $|r(t)|\rightarrow1$, asymptotically losing its time-dependence completely and its decoherence timescale becoming effectively infinite. The central spin seems as if entirely shielded from the bath interactions. However, we also know that in a central spin model, information that spreads to the bath can backflow, which gives this model its characteristic non-Markovian dynamics. Moreover, in the absence of a bath Hamiltonian, the memory of a spin bath is infinite. We suspect that this slowing down of decoherence might be a consequence of the effect of the presence of an unbroken \pt-symmetry on the information backflow.

\section{\label{Heading4}\pt-symmetry Effects on Einselection and Pointer States}

In this section, we will look at the purity of an arbitrary Bloch vector of the central spin evolving under the Hamiltonian in equation (\ref{Htot}) to see whether the slowing down of decoherence might be hinting at a change in the information backflow. Moreover, we wish to see how \pt-symmetry will affect the einselection facilitated by the interaction Hamiltonian in the presence of a transverse self-Hamiltonian.

It can be observed in equation (\ref{Htot}) that the self-Hamiltonian and interaction Hamiltonian are incompatible. In contrast, $V_{\mathcal{E}}$ commutes with the total Hamiltonian. The time-evolution operator generated by the total Hamiltonian can then be expanded into
\begin{equation}
    \hat{U}(t)=\sum_{n=0}^{2^{N}-1} \hat{U}_{B_{n}}\otimes \ket{n}\bra{n},
\end{equation}
where $\hat{U}_{B_{n}}=I\cos{\left(\Omega_{n}t\right)}-i\frac{\left(B_{n}\hat{\sigma}_{0}^{z}+\omega_{0}\hat{\sigma}_{0}^{x}\right)}{\Omega_{n}}\sin{\left(\Omega_{n}t\right)}$ and $\Omega_{n}=B_{n}^{2}+\omega_{0}^{2}$.

For any arbitrary initial Bloch vector $\vec{p}(0)$, we will denote 
\begin{equation}
    \vec{p}(t,B)=\hat{U}_{B}(t)\vec{p}(0)\hat{U}_{B}^{\dagger}(t),
\end{equation}
so that the partial trace operation can be condensed,
\begin{equation}
    \vec{p}(t)=\int \vec{p}(t,B)\Phi(B)dB.
    \label{reduced_Bloch_vector}
\end{equation}
From here, we use an arbitrary initial pure Bloch vector $\vec{p}(0)$,
\begin{equation}
    \vec{p}(0)=p_{x}(0)\hat{\sigma}_{0}^{x}+p_{y}(0)\hat{\sigma}_{0}^{y}+p_{z}(0)\hat{\sigma}_{0}^{z}. 
    \label{arbitrary_initial_Bloch_vector}
\end{equation}
These components represent the expectation value of their corresponding spin operator, which can be interpreted as the spin-polarization about that axis. We can now apply the time-evolution operator on $\vec{p}(0)$ to get the components of $\vec{p}(t,B)$,
\begin{widetext}
\begin{subequations}
\begin{align}
    p_{x}(t,B)&=p_{x}(0)\frac{\omega_{0}^{2}+B^{2}\cos{\left(2\Omega_{B}t\right)}}{\Omega_{B}^{2}}-p_{y}(0)\frac{B}{\Omega_{B}}\sin{\left(2\Omega_{B}t\right)}+p_{z}(0)\frac{2B\omega_{0}}{\Omega_{B}^{2}}\sin^{2}{\left(\Omega_{B}t\right)}, \\
    p_{y}(t,B)&=p_{y}(0)\cos{\left(2\Omega_{B}t\right)}+\frac{\sin{\left(2\Omega_{B}t\right)}}{\Omega_{B}}\left(B p_{x}(0)-\omega_{0}p_{z}(0)\right), \\
    p_{z}(t,B)&=p_{x}(0)\frac{2B\omega_{0}}{\Omega_{B}^{2}}\sin^{2}{\left(\Omega_{B}t\right)}+p_{y}(0)\frac{\omega_{0}}{\Omega_{B}}\sin{\left(2\Omega_{B}t\right)}+p_{z}(0)\frac{B^{2}+\omega_{0}^{2}\cos{\left(2\Omega_{B}t\right)}}{\Omega_{B}^{2}}.
\end{align}
\end{subequations}
\end{widetext}

The exact Bloch vector cannot be solved analytically since $\vec{p}(t,B)$ contains terms with both polynomial and oscillatory functions, which are to be integrated under $\Phi(B)$. As we are mainly concerned about the pointer states, we will be employing the long-time limit. This long-time limit will be further divided into two distinct regimes, which are the strong environment regime ($\omega_{0}\ll s_{N}$) and the strong self-dynamics regime ($\omega_{0}\gg s_{N}$).

\subsection{Strong Environment Regime ($\omega_{0}\ll s_{N}$)}
We can employ the stationary phase approximation to derive the Bloch vector. The resulting Bloch vector components are
\begin{widetext}
\begin{subequations}
\label{strong_environment_Bloch_vector}
\begin{align}
    p_{x}(t)&= p_{x}(0)\left[\Xi{\left(\frac{\omega_{0}}{\sqrt{2(1-\gamma^{2})}s_{N}}\right)}+\frac{1}{\sqrt{8\omega_{0}s_{N}^{2}\left(1-\gamma^{2}\right)t^{3}}}\cos{\left(2\omega_{0}t+\frac{3\pi}{4}\right)}\right], \\
    p_{y}(t)&= \sqrt{\frac{\omega_{0}}{2s_{N}^{2}\left(1-\gamma^{2}\right)t}}\left[p_{y}(0)\cos{\left(2\omega_{0}t+\frac{\pi}{4}\right)-p_{z}(0)\sin{\left(2\omega_{0}t+\frac{\pi}{4}\right)}}\right], \\
    \nonumber
    p_{z}(t)&= p_{z}(0)\left[1-\Xi{\left(\frac{\omega_{0}}{\sqrt{2(1-\gamma^{2})}s_{N}}\right)}+\sqrt{\frac{\omega_{0}}{2s_{N}^{2}\left(1-\gamma^{2}\right)t}}\cos{\left(2\omega_{0}t+\frac{\pi}{4}\right)}\right] \\ 
    &+p_{y}(0)\sqrt{\frac{\omega_{0}}{2s_{N}^{2}\left(1-\gamma^{2}\right)t}}\sin{\left(2\omega_{0}t+\frac{\pi}{4}\right)}, 
\end{align}
\end{subequations}
\end{widetext}
where $\Xi(x)=\sqrt{\pi}xe^{x^{2}}\mathrm{erfc}(x)$. 

We can see in Figure \ref{strong_environment_regime} that as $\gamma\rightarrow1$, the steady-state value of $p_{x}(t)$ and the oscillation amplitude increase. On the other hand, the steady-state value of $p_{z}(t)$ decreases while its oscillation amplitude increases as $\gamma\rightarrow1$. The plot for $p_y(t)$ is not shown as it does not compete in the einselection, meaning that the eigenstates of $\hat{\sigma}_{0}^{y}$ will always decohere to the chosen pointer states for all values of $\gamma$ except at the \pt-transition point.

The oscillation here is the attempt of the system to align to a specific polarization axis in the presence of the total Hamiltonian. As the environment interaction is dominant, $p_{z}(t)$ experiences less pronounced decay in its oscillations. However, the interaction weakens as $\gamma\rightarrow1$, allowing all states to experience coherent oscillations despite the presence of environment interactions.

\begin{figure*}[t]
    \centering
    \includegraphics[width=0.48\linewidth]{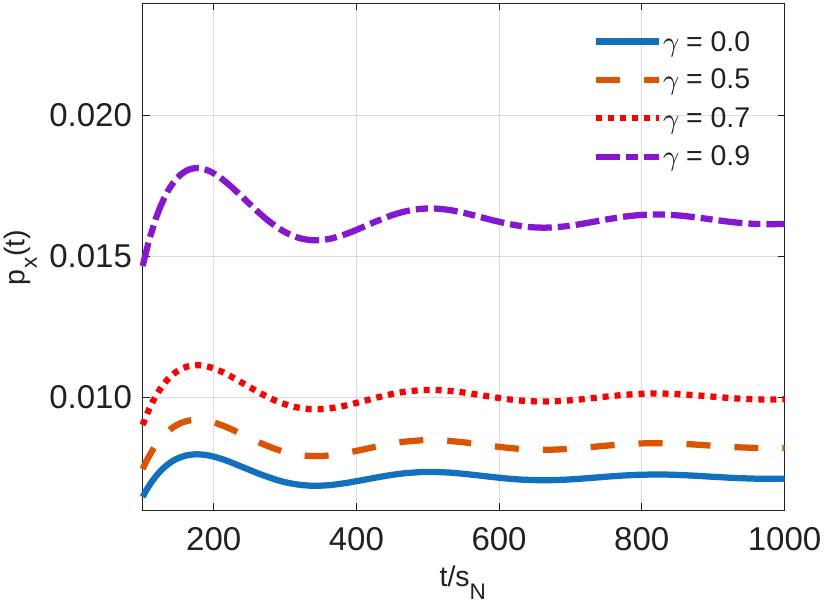}
    \includegraphics[width=0.48\linewidth]{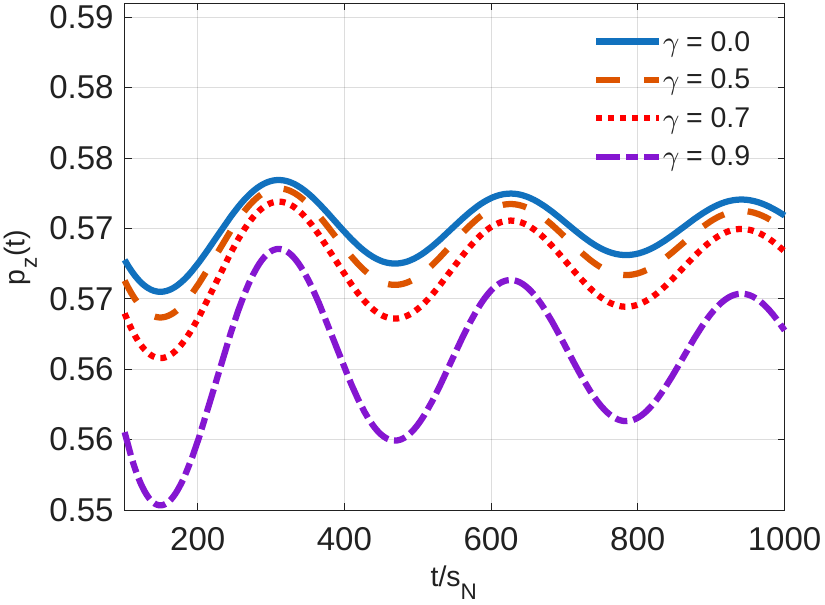}
    \caption{Comparison of Bloch vector components $p_{x}(t)$ and $p_{z}(t)$ for different values of $\gamma$ when $\omega_{0}/s_{N}=0.01$, and $p_{x}(0)=p_{y}(0)=p_{z}(0)=1/\sqrt{3}$. These plots are obtained from equation (\ref{strong_environment_Bloch_vector}).}
    \label{strong_environment_regime}
\end{figure*}

\subsection{Strong Self-Dynamics Regime ($\omega_{0}\gg s_{N}$)}
We can use the expansion $\Omega=\omega_{0}+\frac{B^{2}}{\omega_{0}}$, which will yield
\begin{widetext}
\begin{subequations}
\label{strong_central_spin_Bloch_vector}
\begin{align}
    p_{x}(t)&=p_{x}(0)\left[\Xi{\left(\frac{\omega_{0}}{\sqrt{2(1-\gamma^{2})}s_{N}}\right)}+\frac{s_{N}^{2}(1-\gamma^{2})}{\omega_{0}^{2}}\frac{\cos{\left(2\omega_{0}t+\frac{3}{2}\arctan{\frac{2s_{N}^{2}\left(1-\gamma^{2}\right)t}{\omega_{0}}}\right)}}{\left(1+\left(\frac{2s_{N}^{2}\left(1-\gamma^{2}\right)t}{\omega_{0}}\right)^{2}\right)^{\frac{3}{4}}}\right], \\
    p_{y}(t)&=\left[1+\left(\frac{2s_{N}^{2}\left(1-\gamma^{2}\right)t}{\omega_{0}}\right)^{2}\right]^{-\frac{1}{4}} \left[p_{y}(0)\cos{\left(2\omega_{0}t+\frac{1}{2}\arctan{\frac{2s_{N}^{2}\left(1-\gamma^{2}\right)t}{\omega_{0}}}\right)} \right. \nonumber \\
    & \left. -p_{z}(0)\sin{\left(2\omega_{0}t+\frac{1}{2}\arctan{\frac{2s_{N}^{2}\left(1-\gamma^{2}\right)t}{\omega_{0}}}\right)}\right], \\
    p_{z}(t)&=p_{z}(0)\left[1-\Xi{\left(\frac{\omega_{0}}{\sqrt{2(1-\gamma^{2})}s_{N}}\right)}+\frac{\cos{\left(2\omega_{0}t+\frac{1}{2}\arctan{\frac{2s_{N}^{2}\left(1-\gamma^{2}\right)t}{\omega_{0}}}\right)}}{\left(1+\left(\frac{2s_{N}^{2}\left(1-\gamma^{2}\right)t}{\omega_{0}}\right)^{2}\right)^{\frac{1}{4}}}\right] \nonumber \\
    &+p_{y}(0)\frac{\sin{\left(2\omega_{0}t+\frac{1}{2}\arctan{\frac{2s_{N}^{2}\left(1-\gamma^{2}\right)t}{\omega_{0}}}\right)}}{\left(1+\left(\frac{2s_{N}^{2}\left(1-\gamma^{2}\right)t}{\omega_{0}}\right)^{2}\right)^{\frac{1}{4}}}.
\end{align}
\end{subequations}
\end{widetext}
In contrast to the previous regime, the steady-state value of $p_{x}(t)$ increases as $\gamma\rightarrow1$, while the steady-state value of $p_{z}(t)$ is fixed at $0$, which can be seen in Figure \ref{strong_central_spin_regime}\footnote{A caveat must be taken here that the numerical value of $\omega_{0}/s_{N}$ is technically not valid. However, the value used does not affect the qualitative behavior of the plot. Actual valid values for this approximation exhibit long decay times and rapid oscillations.}.

In this regime, the oscillation is the attempt of the self-Hamiltonian to decouple the interaction from the central spin. The fast oscillations average out the dephasing noise that the central spin experiences from its environment. The pronounced decay in the oscillation amplitude in both $p_{x}(t)$ and $p_{z}(t)$ reflects the pointer states of the system, as their axes of symmetry are along their steady-state value.

\begin{figure*}[t]
    \centering
    \includegraphics[width=0.48\linewidth]{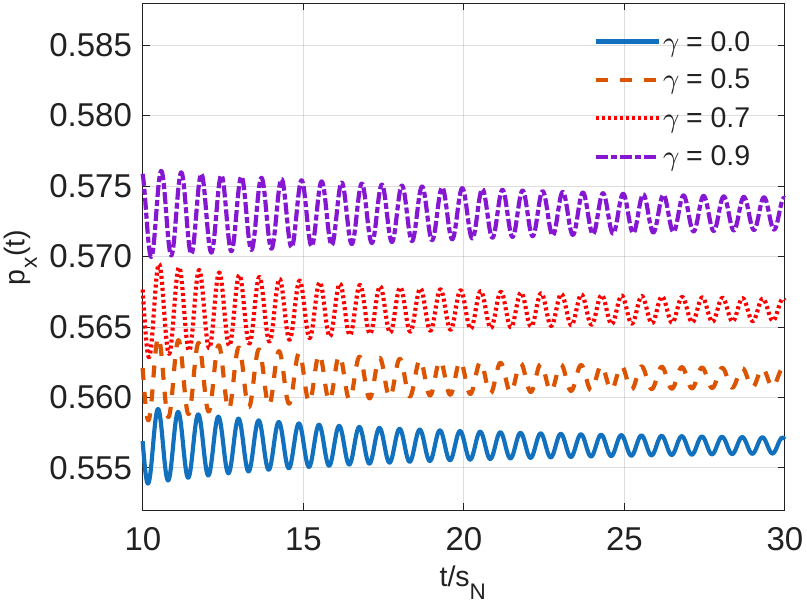}
    \includegraphics[width=0.48\linewidth]{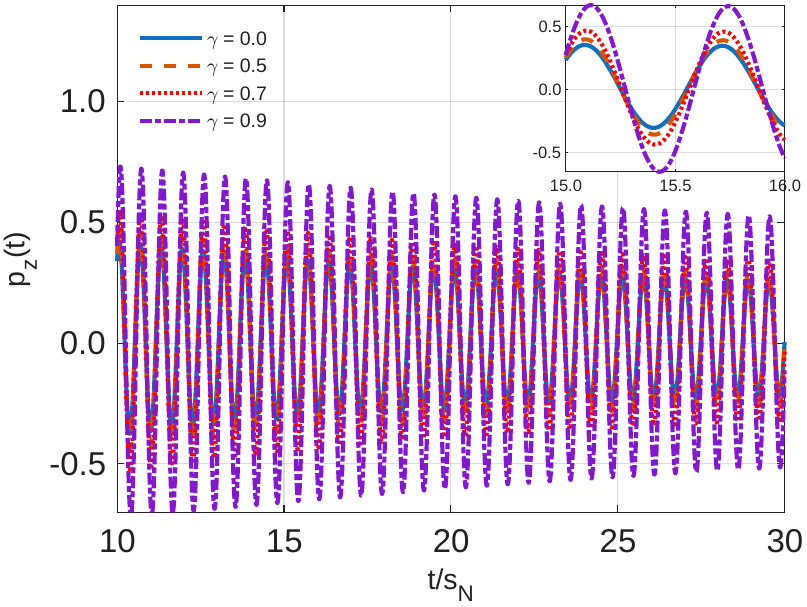}
    \caption{Comparison of Bloch vector components $p_{x}(t)$ and $p_{z}(t)$ for different values of $\gamma$ when $\omega_{0}/s_{N}=5$, and $p_{x}(0)=p_{y}(0)=p_{z}(0)=1/\sqrt{3}$. These plots are obtained from equation (\ref{strong_central_spin_Bloch_vector}).}
    \label{strong_central_spin_regime}
\end{figure*}

\subsection{Pointer States}
Investigating these two limits, we will be able to see that these are almost the same long-time limits as seen in \cite{Cucchietti}, except for the rescaled terms. This is a relic of the effective coupling affecting the variance of the characteristic function. Looking at the resulting Bloch vector components, the function $\Xi(x)$ is present in both limits. This function has the following limits,
\begin{align}
    &\lim_{x\rightarrow0^{+}} \Xi(x)\rightarrow0, \label{strong_environment} \\ 
    &\lim_{x\rightarrow+\infty} \Xi(x)\rightarrow1. \label{strong_self}
\end{align}
The steady-state values for both limits, i.e., $t\rightarrow\infty$, are
\begin{align}
    p_{x}(t)&= p_{x}(0)\Xi{\left(\frac{\omega_{0}}{\sqrt{2(1-\gamma^{2})}s_{N}}\right)}, \\
    p_{y}(t)&= 0,\\
    p_{z}(t)&= p_{z}(0)\left[1-\Xi{\left(\frac{\omega_{0}}{\sqrt{2(1-\gamma^{2})}s_{N}}\right)}\right].
\end{align}

The time-independent function $\Xi(x)$ determines the steady-state value of $p_{x}(t)$ and $p_{z}(t)$, which also decides which pointer states are selected. For intermediate values of $\gamma$, $\omega_{0}/s_{N}$ competes with $\gamma$ to settle which pointer state the system will choose. Nevertheless, $\Xi(x)$ collapses to the limit in equation (\ref{strong_self}) when $\gamma\rightarrow1$. This forces the system to choose the eigenstates of the self-Hamiltonian as its pointer states, regardless of the regime the system is in. This is due to the apparent vanishing of the interaction Hamiltonian at the \pt-transition point. This is similar to how dynamical decoupling protects the central spin states from dephasing due to interactions \cite{Viola_Knill_Lloyd_1999}. The difference in this case is that a tuning parameter in the interaction Hamiltonian drives this decoupling.

In contrast, in dynamical decoupling, any external field perpendicular to the interaction axis is applicable. Now, one might ask, "What if $\hat{\sigma}_{0}^{x}$ in the self-Hamiltonian is replaced by $\hat{\sigma}_{0}^{y}$?" The long-time limit of the Bloch vector components will change such that
\begin{align}
    p_{x}(t)&= 0, \\
    p_{y}(t)&= p_{y}(0)\Xi{\left(\frac{\omega_{0}}{\sqrt{2(1-\gamma^{2})}s_{N}}\right)},\\
    p_{z}(t)&= p_{z}(0)\left[1-\Xi{\left(\frac{\omega_{0}}{\sqrt{2(1-\gamma^{2})}s_{N}}\right)}\right].
\end{align}
Effectively, this means that the system chooses between the eigenstates of $\hat{\sigma}_{0}^{y}$ and $\hat{\sigma}_{0}^{z}$. However, the corresponding pseudo-Hermitian observable for $\hat{\sigma}_{0}^{y}$ is
\begin{equation}
    \hat{\sigma}_{0,pH}^{y}=T^{-1}\hat{\sigma}_{0}^{y}T=\frac{1}{\sqrt{1-\gamma^{2}}}\hat{\sigma}_{0}^{y}-\frac{i\gamma}{\sqrt{1-\gamma^{2}}}\hat{\sigma}_{0}^{z},
\end{equation}
which makes the construction of experimental design for such a system even more complicated. However, it is still a feasible self-Hamiltonian, given that it is pseudo-Hermitian with respect to the metric operator.

\subsection{Purity}
The central spin Bloch vector is directly related to the purity $P$ of its states, through the relation $P=\frac{1}{2}(1+|\vec{p}|^{2})$. The plot for the purity for each regime is shown in Figure \ref{Purity}. The oscillations in both regimes suggest the non-Markovian behavior of the central spin model. This non-Markovianity will be quantified in the succeeding section. For the meantime, we investigate the purity of the states under the influence of the Hamiltonian. We know that purity is a measure of uncertainty about a specific state of a system. As interaction causes decoherence, and consequently dissipates information from the central spin to its environment, purity decreases.

\begin{figure*}
    \centering
    \begin{subfigure}[t]{0.48\linewidth}
        \centering
        \includegraphics[width=\linewidth]{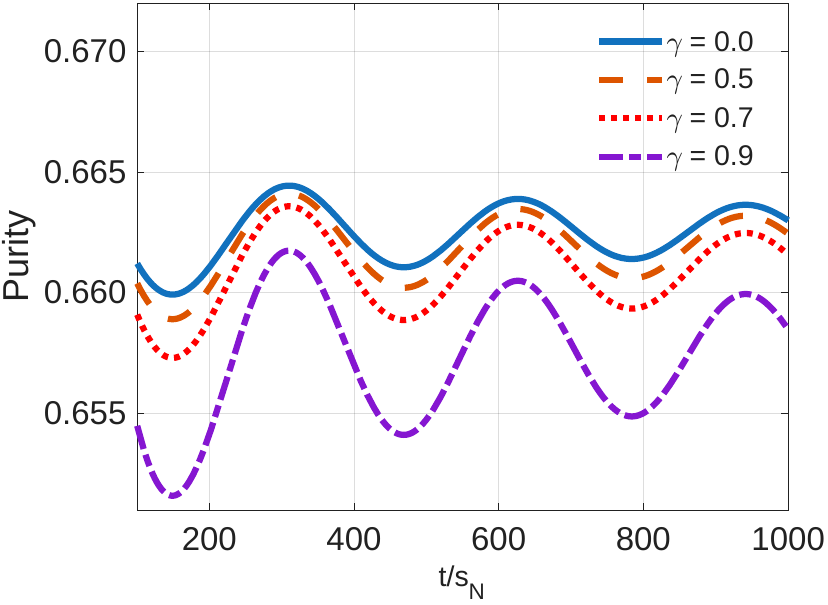}
        \caption{}
        \label{Strong_Environment_Purity}
    \end{subfigure}
    \hfill
    \begin{subfigure}[t]{0.48\linewidth}
        \centering
        \includegraphics[width=\linewidth]{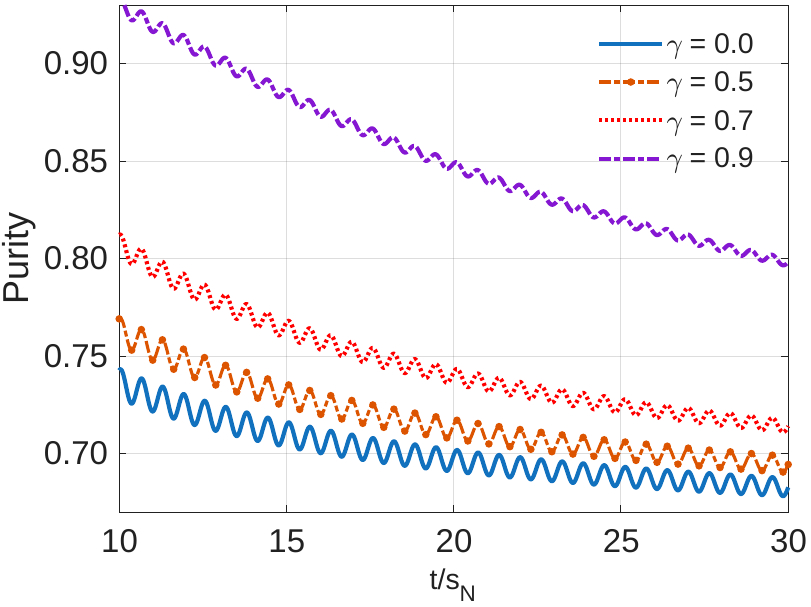}
        \caption{}
        \label{Strong_Central_spin_Purity}
    \end{subfigure}
    \caption{Purity of the central spin states of the system in Figure (a) \ref{strong_environment_regime} and (b) Figure \ref{strong_central_spin_regime}, respectively.}
    \label{Purity}
\end{figure*}

However, it is rather counterintuitive to see that the steady-state value of the purity decreases as $\gamma$ increases in the strong environment regime, as can be seen in Figure \ref{Strong_Environment_Purity}. This is because in section \ref{Heading3}, we have shown that increasing $\gamma$ suppresses decoherence. If so, will this contradiction persist until the \pt-transition point? In Figure \ref{Purity_vs_Gamma}, we see that there is a turning point at $(1-\gamma,P)=(1.33\times10^{-4},7/12)$, from where the steady-state purity increases as $\gamma$ increases. This is interpreted as the point where the self-Hamiltonian finally overpowers the dephasing noise from the environment. This point also represents the value of $\gamma$ that produces the highest dissipation of information to the environment. In terms of the Bloch vector components, this is where $p_{z}(t\rightarrow\infty)=p_{x}(t\rightarrow\infty)$. Starting from this point, $\gamma$ starts forcing the system to choose the eigenstates of the self-Hamiltonian as pointer states. We can then divide purity into three distinct regimes depending on which Bloch vector component dominates the competition for the selection of pointer states, which are shown in Figure \ref{Purity_vs_Gamma}. This shows that the competition between the self-Hamiltonian and interaction Hamiltonian is not as trivial as just the interaction weakening as $\gamma\rightarrow1$. This provides a novel approach for engineering spin systems that are robust under environmental interactions.

\subsection{\pt-transition point}
Neither of the limits employed in this section includes the \pt-transition point, since the metric operator is undefined at this point. However, we can still take the limit of the Bloch vector at this point, since at the \pt-transition point, the characteristic function becomes a Dirac delta function. This will allow for the closed form of the Bloch vector to be analytically solved without having to employ any approximations. The Bloch vector components at the \pt-transition point become
\begin{subequations}
\begin{align}
    p_{x}(t)&=p_{x}(0), \\
    p_{y}(t)&=p_{y}(0)\cos{\left(2\omega_{0}t\right)}-p_{z}(0)\sin{\left(2\omega_{0}t\right)}, \\
    p_{z}(t)&=p_{y}(0)\sin{\left(2\omega_{0}t\right)}+p_{z}(0)\cos{\left(2\omega_{0}t\right)}.
\end{align}
\end{subequations}
This is exactly the dynamics of the Pauli spin operators that is generated by the self-Hamiltonian in the absence of the interaction Hamiltonian. The central spin will experience Larmor precession along the $x$-axis, which will keep the purity of the states conserved indefinitely. Notice that the precession disappears if $\omega_{0}=0$. This complete suppression of decoherence and internal dynamics is the signal of an exceptional point for pseudo-Hermitian systems, where eigenvalues and eigenvectors coalesce, effectively rendering the Hamiltonian defective.

\section{\label{Heading5}\pt-symmetry Effects to Information Backflow and Non-Markovian Characteristic}

The change in the purity revival shown in the previous section raises the question of how \pt-symmetry affects the information backflow in the central spin model. The central spin model is known to exhibit non-Markovian behavior because of the memory retention of the spin bath, regardless of whether it is static or possesses internal dynamics \cite{Coish_Loss_2004, Barnes_Cywinski_Das_Sarma_2012, Coish_Fischer_Loss_2010}.

The information backflow to the central spin, which causes its non-Markovian behavior, is characterized by the increase in the trace distance with respect to time, i.e.,
\begin{equation}
    \frac{d}{dt}D(\hat{\rho}_{1}(t),\hat{\rho}_{2}(t))>0.
\end{equation}
To quantify its non-Markovianity, we can use the Breuer–Laine–Piilo (BLP) measure, defined as 
\begin{equation}
    \mathcal{N}_{BLP}=\max_{\hat{\rho}_{1,2}(0)}\int_{\dot{D}(t)>0}dt \ \dot{D}(t).
    \label{BLP_measure}
\end{equation}
This measure was chosen specifically because of its high sensitivity even in weak coupling regimes, where the system is headed as $\gamma\rightarrow1$ \cite{Rashid_Mala_Bashir_Lone_2025}. Moreover, BLP distinguishes between pure decoherence and reversible phenomena due to the memory of a specific environment through the trace distance. By choosing a pair of orthogonal initial density matrices, we can ensure that the BLP measure obtained is indeed the maximum for any pair of states, as orthogonal pairs are the optimal pairs for quantifying non-Markovianity \cite{Wissmann_Karlsson_Laine_Piilo_Breuer_2012}. Equation (\ref{BLP_measure}) then becomes
\begin{equation}
    \mathcal{N}_{BLP}=\sum_{n}\left[D(\hat{\rho}_{1}(t_{n+1}),\hat{\rho}_{2}(t_{n+1}))-D(\hat{\rho}_{1}(t_{n}),\hat{\rho}_{2}(t_{n}))\right],
\end{equation}
where $(t_{n},t_{n+1})$ are the time-intervals when the trace distance is increasing.

Upon the use of BLP measure, we will quantify the effect of \pt-symmetry on the non-Markovianity. Note that the trace distance is generally not preserved under non-unitary transformations. Nevertheless, it can be proven that under the modified adjoint and inner product, the Dyson map preserves the trace distance (see Theorem \ref{trace_distance_invariance} in Appendix \ref{app:pseudo} for the proof).

We will be taking the trace distance of the density matrices of two arbitrary orthogonal pure states evolving under the influence of the Hamiltonian in equation (\ref{Htot}). We let $\hat{\rho}_{1}(0)$ and $\hat{\rho}_{2}(0)$ be these density matrices. Then, their respective Bloch vector representations $\vec{p}_{1}(0)$ and $\vec{p}_{2}(0)$ are related such that 
\begin{equation}
    \vec{p}_{1}(0)=-\vec{p}_{2}(0).
    \label{orthogonality}
\end{equation}
The reduced density matrices corresponding to these initial density matrices are
\begin{equation}
    \hat{\rho}_{1}(t)=\int \hat{U}_{B}(t)\hat{\rho}_{1}(0)\hat{U}_{B}^{\dagger}(t)\Phi(B)dB
\end{equation}
and
\begin{equation}
    \hat{\rho}_{2}(t)=\int \hat{U}_{B}(t)\hat{\rho}_{2}(0)\hat{U}_{B}^{\dagger}(t)\Phi(B)dB,
\end{equation}
respectively. Then,
\begin{equation}
    \Delta\hat{\rho}(t)=\int \hat{U}_{B}(t)(\hat{\rho}_{1}(0)-\hat{\rho}_{1}(0))\hat{U}_{B}^{\dagger}(t)\Phi(B)dB.
\end{equation}
Decomposing the density matrices into their respective Bloch vector decompositions and substituting equation (\ref{orthogonality}),
\begin{equation}
    \Delta\hat{\rho}(t)=\int \hat{U}_{B}(t)\vec{p}_{1}(0)\hat{U}_{B}^{\dagger}(t)\Phi(B)dB=\vec{p}_{1}(t)
\end{equation}
This is exactly what is stated in equation (\ref{reduced_Bloch_vector}) for arbitrary Bloch vectors. Consequently, 
\begin{equation}
    D(\hat{\rho}_{1}(t),\hat{\rho}_{2}(t))=\frac{1}{2}\mathrm{Tr}\sqrt{\Delta\hat{\rho}(t)^{\dagger}\Delta\hat{\rho}(t)}.
\end{equation}
Since $\Delta\hat{\rho}(t)^{\dagger}=\Delta\hat{\rho}(t)$, then
\begin{align}
    D(\hat{\rho}_{1}(t),\hat{\rho}_{2}(t))&=\frac{1}{2}\mathrm{Tr}\sqrt{p_{1}^{2}(t)\mathbb{I}} \nonumber \\
    &=\sqrt{p_{1,x}^{2}(t)+p_{1,y}^{2}(t)+p_{1,z}^{2}(t)}.
\end{align}
The trace distance can be derived analytically using the Bloch vector components given in equations (\ref{strong_environment_Bloch_vector}) and (\ref{strong_central_spin_Bloch_vector}).

\begin{figure}[t]
    \centering
    \includegraphics[width=\linewidth]{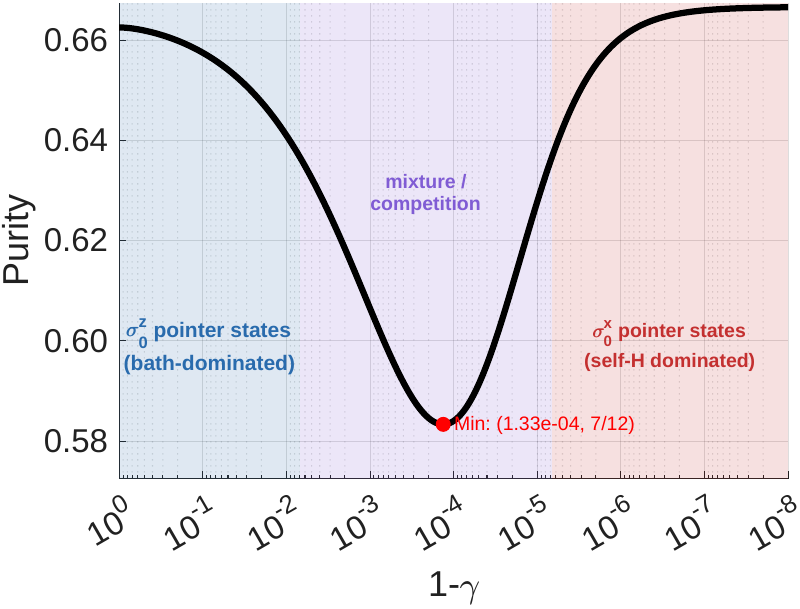}
    \caption{Purity vs. $1-\gamma$ for the system described in Figure \ref{strong_environment_regime}.}
    \label{Purity_vs_Gamma}
\end{figure}

The intervals where trace distance increases cannot be determined analytically, as the oscillations are aperiodic. As such, these intervals are obtained numerically. The BLP measure is also unbounded, especially for models like the central spin systems where information backflow persists indefinitely, albeit in very small amounts after long timescales. Moreover, equations (\ref{strong_environment_Bloch_vector}) and (\ref{strong_central_spin_Bloch_vector}) are derived by utilizing approximations, so their validity is directly tied to the time they will be evaluated. As such, we will limit the time interval used to calculate the BLP measure to ensure a finite value, as well as to ensure the positivity of the density matrix.

The plots of the BLP measure for systems in the strong environment regime and strong self-dynamics regime are in Figure \ref{Strong_Environment_BLP} and Figure \ref{Strong_Central_spin_BLP}, respectively. The former was obtained within the time interval $[2.5\times10^{5},10^{6}] \ t/s_{N}$ because of the very slow oscillations in the trace distance in this regime. This specific choice of time interval is to assure positivity throughout all value of $\gamma$, as evaluating at $t$ less than $2.5\times10^{5}$ can break positivity for some values of $\gamma$ near $1-\left[1\times10^{-8}\right]$. Moreover, the oscillation amplitude in this regime decays more slowly as $\gamma\rightarrow1$. Meanwhile, the latter was obtained within the time interval $[10,10^{2}] \ t/s_{N}$. The trace distance in this regime is characterized by rapid oscillations and sudden decay in the oscillation amplitude. As such, the BLP measure for this regime was obtained in a smaller time interval.

The observed dip in the purity in the strong environment regime from the previous section strongly suggested that the information backflow does not abruptly change as $\gamma\rightarrow1$. We have also investigated the plot of the BLP measure against $\gamma$, which is shown in Figure \ref{Strong_Environment_BLP}. The BLP measure has shown a peak at $(1-\gamma,\mathcal{N}_{BLP})=(1.33\times10^{-4},17.1)$, which is near the turning point of the purity. This strengthens our claim in the previous section that the competition between the central spin Hamiltonian and the interaction Hamiltonian is not as trivial as just the interaction vanishing as $\gamma\rightarrow1$. The information backflow in the strong environment regime changes smoothly as $\gamma$ approaches that specific turning point, which signifies the maximum information backflow. From this point on, \pt-symmetry suppresses decoherence and shields the central spin from the dephasing noise that environmental interaction introduces. We can then divide the BLP measure in this regime into two distinct regimes depending on which decoherence process follows, as can be seen in Figure \ref{Strong_Environment_BLP}.

\begin{figure*}
    \centering
    \begin{subfigure}[t]{0.48\linewidth}
        \centering
        \includegraphics[width=\linewidth]{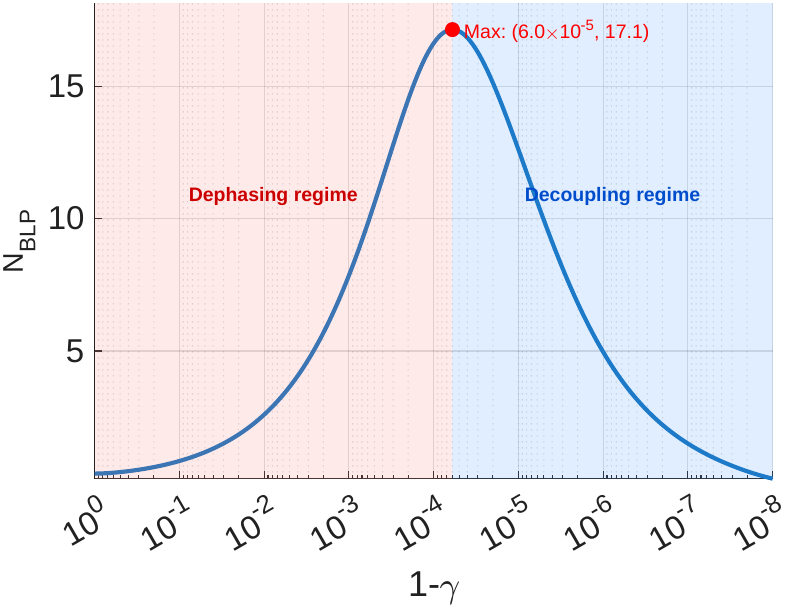}
        \caption{}
        \label{Strong_Environment_BLP}
    \end{subfigure}
    \hfill
    \begin{subfigure}[t]{0.48\linewidth}
        \centering
        \includegraphics[width=\linewidth]{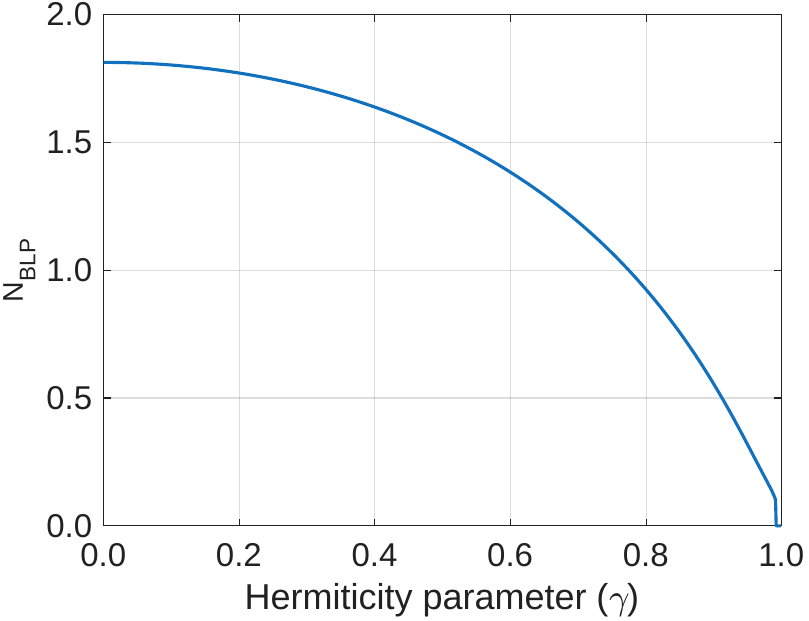}
        \caption{}
        \label{Strong_Central_spin_BLP}
    \end{subfigure}
    \caption{$\mathcal{N}_{\mathrm{BLP}}$ vs. $1-\gamma$ for the system in (a) Figure \ref{strong_environment_regime} and $\mathcal{N}_{\mathrm{BLP}}$ vs. $\gamma$ for the system in (b) Figure \ref{strong_central_spin_regime}.}
    \label{BLP_measure}
\end{figure*}

On the other hand, in the strong self-dynamics regime, the BLP measure monotonically decreases. The weakening of the interaction Hamiltonian as $\gamma\rightarrow1$ allows the self-Hamiltonian to average out the dephasing noise more effectively. This leads to less information spreading into the environment, consequently resulting in less information backflow, as can be seen in Figure \ref{Strong_Central_spin_BLP}. Approaching the \pt-transition point, the interaction effectively vanishes, allowing $\omega_{0}$ to dictate the dynamics of the system. Furthermore, all dynamics freeze out when approaching the exceptional point, where $\omega_{0}=0$, completely shielding the central spin states from noise that can cause decoherence.

\section{\label{Conclusion}Conclusions and Recommendations}

We have applied \pt-symmetry to a central spin model through a \pt-symmetric interaction Hamiltonian. The decoherence of the central spin slowed down in the unbroken \pt-symmetry regime and freezes at the \pt-transition point. The \pt-symmetric interaction Hamiltonian was modified to include a self-Hamiltonian that is pseudo-Hermitian with respect to the metric operator. We have shown that \pt-symmetry can affect the einselection by forcing the system to select the eigenstates of the self-Hamiltonian as $\gamma\rightarrow1$. However, in the strong environment regime, this competition is not trivial. A paradoxical decrease in the purity happens as $\gamma\rightarrow1$. This is contradictory to our result that decoherence slows down as $\gamma\rightarrow1$. However, looking more closely at the values extremely close to $\gamma=1$, we are able to identify a turning point, a minimum, which signifies the value of $\gamma$ that allows the maximum information dissipation to the environment. From that point, the pointer states start to become the eigenstates of the self-Hamiltonian. This is further solidified by the BLP measure we have quantified. A similar turning point, this time a maximum, appears in this measure in the same regime. This measure quantifies the information backflow from the environment to the central spin, which means that at this turning point, the information backflow reaches a saturation point. From the turning point, less information backflows to the central spin, which also means that there is less information that spreads into the environment. Therefore, this turning point represents where the self-Hamiltonian starts to dominate the dynamics, shielding the central spin from the dephasing noise brought about by environmental interaction.

Further extensions of this study should involve more complex and/or realistic interactions, such as Heisenberg interaction, Lipkin-Meshkov-Glick interaction, or Dzyaloshinskii–Moriya interaction endowed with \pt-symmetry. The non-commutativity stemming from these complex interactions can introduce profound memory effects. Moreover, if the Hermiticity parameter $ \ \gamma$ is a time-dependent or site-dependent parameter, a more complicated methodology must be considered due to the constraint that is imposed to maintain an unbroken \pt-symmetry. Nevertheless, it is intriguing to explore the potential nontrivial effects that may arise from these dependencies.

Future work should also investigate whether non-Markovianity also arises in other \pt-symmetric decoherence models, such as spin-boson models. It is interesting to explore whether the same turning-point behavior can also account for the slowing down of decoherence observed from these models. Furthermore, non-Hermitian Hamiltonians are effective descriptions derived from larger, Hermitian microscopic Hamiltonians. Then, it should be possible to derive a microscopic Hamiltonian that reduces, under appropriate approximations, to an effective \pt-symmetric central spin model. Establishing this connection would ground the phenomenological \pt-symmetric approach in a first-principles microscopic theory, strengthening its physical justification.

\flushbottom
\appendix

\section{Review on Pseudo-Hermitian Formalism}
\label{app:pseudo}
This appendix summarizes the mathematical framework of pseudo-Hermitian
quantum mechanics used throughout this work. The results presented here
are standard and follow Refs.~\cite{Mostafazadeh2002a,Mostafazadeh2002b}.

\begin{definition}[Pseudo-Hermitian Hamiltonian]
A linear operator $\hat{H}$ is said to be pseudo-Hermitian if there exists
an invertible positive-definite Hermitian operator $\eta$ satisfying
\begin{equation}
    \hat{H}^{\dagger}=\eta\hat{H}\eta^{-1}.
\label{eq:app_pH}
\end{equation}
If $\hat{H}$ additionally satisfies
\begin{equation}
    [\hat{H},\pt]=0,
\end{equation}
then the spectrum is entirely real within the $\pt$-unbroken regime
\cite{BenderBoettcher,Bender2002}.
\end{definition}

\begin{lemma}[Metric Operator {\cite{Mostafazadeh2002a}}]
Suppose that $\hat{H}$ possesses an exact invertible antilinear symmetry.
Then there exists a positive-definite Hermitian metric operator $\eta$
satisfying Eq.~(\ref{eq:app_pH}).
\end{lemma}

\begin{lemma}[Dyson Map {\cite{Mostafazadeh2002b}}]
\label{Lemma_A_2}
Since $\eta$ is positive definite, there exists a unique positive
Hermitian operator
\begin{equation}
    T=\sqrt{\eta}, 
\end{equation}
called the Dyson map, satisfying
\begin{equation}
    T^{\dagger}=T,
\end{equation}
and
\begin{equation}
    T^{2}=\eta.
\end{equation}
\end{lemma}

\begin{theorem}[Hermitian Representation{\cite{Mostafazadeh2002b}}]
\label{Hermitian_Representation}
Let
\begin{equation}
    \hat{h}=T\hat{H}T^{-1}.
\end{equation}
Then
\begin{equation}
    \hat{h}^{\dagger}=\hat{h}.
\end{equation}
\end{theorem}

\begin{definition}[Biorthogonal Basis]
Let $\left\{\ket{\psi_{n}}\right\}$ and $\left\{\ket{\phi_{n}}\right\}$ denote the right and left eigenvectors of $\hat{H}$, respectively. They satisfy
\begin{equation}
    \braket{\phi_m}{\psi_n}=\delta_{mn},
\end{equation}
together with the completeness relation
\begin{equation}
    \sum_{n}\ket{\psi_n}\bra{\phi_n}=\hat{I},
\end{equation}
where
\begin{equation}
    \bra{\phi_n}=\bra{\psi_n}\eta.
\end{equation}
\end{definition}

\begin{definition}[Physical Inner Product]
The physical inner product is defined by
\begin{equation}
    \left(\psi,\phi\right)_{\eta}=\bra{\psi}\eta\ket{\phi},
\end{equation}
which is positive definite,
\begin{equation}
    \bra{\psi}\eta\ket{\psi}>0,
\end{equation}
for every nonvanishing state $\ket{\psi}$.
\end{definition}

\begin{definition}[Pseudo-Hermitian Observable]
Let $\hat{O}$ be a Hermitian observable in the equivalent Hermitian
representation. The corresponding observable in the pseudo-Hermitian
representation is
\begin{equation}
    \hat{O}_{\mathrm{pH}}=T^{-1}\hat{O}T.
\end{equation}
It satisfies
\begin{equation}
    \hat{O}_{\mathrm{pH}}^{\#}=\eta^{-1}\hat{O}_{\mathrm{pH}}^{\dagger}\eta=\hat{O}_{\mathrm{pH}}.
\end{equation}
\end{definition}

\begin{theorem}[Expectation-Value Invariance]
Let
\begin{equation}
    \ket{\chi}=T\ket{\psi}
\end{equation}
denote the state in the Hermitian representation. Then
\begin{equation}
    \bra{\chi}\hat{O}\ket{\chi}=\bra{\psi}\eta\hat{O}_{\mathrm{pH}}\ket{\psi}=\bra{\phi}\hat{O}_{\mathrm{pH}}\ket{\psi}.
\end{equation}
\end{theorem}
\begin{proof}
Using
\begin{equation}
    \ket{\chi}=T\ket{\psi},
\end{equation}
one finds
\begin{equation}
    \bra{\chi}\hat{O}\ket{\chi}=\bra{\psi}T\hat{O}T\ket{\psi}.
\end{equation}
Since
\begin{equation}
    \hat{O}_{\mathrm{pH}}=T^{-1}\hat{O}T,
\end{equation}
it follows that
\begin{equation}
    T\hat{O}T=T^{2}\hat{O}_{\mathrm{pH}}=\eta\hat{O}_{\mathrm{pH}}.
\end{equation}
Therefore,
\begin{equation}
    \bra{\chi}\hat{O}\ket{\chi}=\bra{\psi}\eta\hat{O}_{\mathrm{pH}}\ket{\psi}.
\end{equation}
Using
\begin{equation}
    \bra{\phi}=\bra{\psi}\eta,
\end{equation}
one obtains
\begin{equation}
    \bra{\chi}\hat{O}\ket{\chi}=\bra{\phi}\hat{O}_{\mathrm{pH}}\ket{\psi},
\end{equation}
which completes the proof.
\end{proof}

\begin{theorem}[Preservation of the Physical Inner Product{\normalfont\cite{Mostafazadeh2002b,MOSTAFAZADEH_2010}}]
Let
\begin{equation}
    \ket{\chi_{n}}=T\ket{\psi_{n}},
\end{equation}
where $T=\sqrt{\eta}$ denotes the Dyson map. Then
\begin{equation}
    \braket{\chi_{m}}{\chi_{n}}=\bra{\psi_{m}}\eta\ket{\psi_{n}}.
\end{equation}
Consequently, the orthonormality of the Hermitian representation,
\begin{equation}
    \braket{\chi_{m}}{\chi_{n}}=\delta_{mn},
\end{equation}
is equivalent to the orthonormality of the pseudo-Hermitian representation,
\begin{equation}
    \bra{\psi_{m}}\eta\ket{\psi_{n}}=\delta_{mn}.
\end{equation}
\end{theorem}
\begin{proof}
Using the Dyson transformation,
\begin{equation}
    \ket{\chi_n}=T\ket{\psi_n},
\end{equation}
one finds
\begin{equation}
    \braket{\chi_m}{\chi_n}=\bra{\psi_m}T^{\dagger}T\ket{\psi_n}.
\end{equation}
Since
\begin{equation}
    T^{\dagger}=T,
\end{equation}
and
\begin{equation}
    T^{2}=\eta,
\end{equation}
it follows immediately that
\begin{equation}
    \braket{\chi_m}{\chi_n}=\bra{\psi_m}\eta\ket{\psi_n}.
\end{equation}
Therefore,
\begin{equation}
    \braket{\chi_m}{\chi_n}=\delta_{mn}
\end{equation}
if and only if
\begin{equation}
    \bra{\psi_m}\eta\ket{\psi_n}=\delta_{mn},
\end{equation}
which completes the proof.
\end{proof}

\begin{theorem}[Pseudo-Unitary Time Evolution{\normalfont\cite{Mostafazadeh2002b,MOSTAFAZADEH_2010}}]
Let
\begin{equation}
    \hat{U}(t)=e^{-i\hat{H}t}
\end{equation}
be the time-evolution operator generated by the pseudo-Hermitian Hamiltonian $\hat{H}$. Then
\begin{equation}
    \hat{U}^{\dagger}(t)\eta\hat{U}(t)=\eta.
\label{eq:pseudounitary}
\end{equation}
Furthermore, if
\begin{equation}
    \hat{h}=T\hat{H}T^{-1},
\end{equation}
then the corresponding Hermitian evolution operator
\begin{equation}
    \hat{u}(t)=e^{-i\hat{h}t}
\end{equation}
satisfies
\begin{equation}
    \hat{u}(t)=T\hat{U}(t)T^{-1}.
\end{equation}
\end{theorem}
\begin{proof}
Expanding the exponential,
\begin{equation}
    \hat{u}(t)=\sum_{n=0}^{\infty}\frac{(-it)^{n}}{n!}\hat{h}^{n}.
\end{equation}
Since
\begin{equation}
    \hat{h}=T\hat{H}T^{-1},
\end{equation}
one has
\begin{equation}
    \hat{h}^{n}=T\hat{H}^{n}T^{-1},
\end{equation}
for every non-negative integer $n$. Hence,
\begin{equation}
    \hat{u}(t)=T\left(\sum_{n=0}^{\infty}\frac{(-it)^{n}}{n!}\hat{H}^{n}\right)T^{-1},
\end{equation}
which gives
\begin{equation}
    \hat{u}(t)=T\hat{U}(t)T^{-1}.
\end{equation}
Since
\begin{equation}
    \hat{u}^{\dagger}(t)\hat{u}(t)=\hat{I},
\end{equation}
it follows that
\begin{equation}
    \hat{U}^{\dagger}(t)T^{2}\hat{U}(t)=T^{2}.
\end{equation}
Using
\begin{equation}
    T^{2}=\eta,
\end{equation}
one obtains Eq.~(\ref{eq:pseudounitary}).\
\end{proof}

\begin{definition}[Density Operator{\normalfont\cite{MOSTAFAZADEH_2010}}]
The density operator in the pseudo-Hermitian representation is defined by
\begin{equation}
    \hat{\rho}_{\mathrm{pH}}=T^{-1}\hat{\rho}T,
\end{equation}
where $\hat{\rho}$ denotes the density operator in the equivalent Hermitian representation.
\end{definition}

\begin{theorem}[Metric Invariance of the Trace Distance]\label{trace_distance_invariance}
Let
\begin{equation}
    \Delta\hat{\rho}=\hat{\rho}_{1,\mathrm{pH}}-\hat{\rho}_{2,\mathrm{pH}},
\end{equation}
where
\begin{equation}
    \hat{\rho}_{i,\mathrm{pH}}=T^{-1}\hat{\rho}_{i}T,
    \qquad
    i=1,2.
\end{equation}
Then the trace distance satisfies
\begin{equation}
    D_{\eta}\left(\hat{\rho}_{1,\mathrm{pH}},\hat{\rho}_{2,\mathrm{pH}}\right)=D\left(\hat{\rho}_{1},\hat{\rho}_{2}\right),
\end{equation}
where
\begin{equation}
    D_{\eta}=\frac{1}{2}\mathrm{Tr}\left[\sqrt{(\Delta\hat{\rho})^{\#}(\Delta\hat{\rho})}\right].
\end{equation}
\end{theorem}
\begin{proof}
Using the similarity transformation,
\begin{equation}
    \Delta\hat{\rho}=T^{-1}\left(\hat{\rho}_{1}-\hat{\rho}_{2}\right)T,
\end{equation}
together with the pseudo-Hermitian adjoint,
\begin{equation}
    (\Delta\hat{\rho})^{\#}=\eta^{-1}(\Delta\hat{\rho})^{\dagger}\eta,
\end{equation}
one obtains
\begin{equation}
    (\Delta\hat{\rho})^{\#}(\Delta\hat{\rho})=T^{-1}\left(\hat{\rho}_{1}-\hat{\rho}_{2}\right)^{\dagger}\left(\hat{\rho}_{1}-\hat{\rho}_{2}\right)T.
\end{equation}
Since the trace is invariant under similarity transformations,
{\small
\begin{equation}
    \mathrm{Tr}\left[\sqrt{(\Delta\hat{\rho})^{\#}(\Delta\hat{\rho})}\right]=\mathrm{Tr}\left[\sqrt{\left(\hat{\rho}_{1}-\hat{\rho}_{2}\right)^{\dagger}\left(\hat{\rho}_{1}-\hat{\rho}_{2}\right)}\right],
\end{equation}}
which immediately yields
\begin{equation}
    D_{\eta}\left(\hat{\rho}_{1,\mathrm{pH}},\hat{\rho}_{2,\mathrm{pH}}\right)=D\left(\hat{\rho}_{1},\hat{\rho}_{2}\right).
\end{equation}
\end{proof}

\section{Derivation of Appropriate $\mathcal{P}$ and $\mathcal{T}$ Operators of $V_{\mathcal{S}}$}
To ensure that $V_{\mathcal{S}}$ is \pt-symmetric, there must exist $\mathcal{P}$ and $\mathcal{T}$ operators such that
\begin{equation}
    [\hat{H},\mathcal{PT}]=0.
\end{equation}
We know that $\mathcal{P}$ is a unitary operator, while $\mathcal{T}$ is an antiunitary operator, often decomposed into
\begin{equation}
    \mathcal{T}=U\mathcal{K},
\end{equation}
where $U$ is a unitary operator and $\mathcal{K}$ is the complex conjugation operator. However, $\mathcal{P}$ and $\mathcal{T}$ do not need to be symmetries of the system individually, since the system is \pt-symmetric. The \pt-symmetric Hamiltonian $V_{\mathcal{S}}$ in this paper is $\hat{\sigma}_{0}^{z}+i\gamma\hat{\sigma}_{0}^{y}$, which is just the well-known \pt-symmetric Hamiltonian $\hat{\sigma}_{0}^{x}+i\gamma\hat{\sigma}_{0}^{z}$ subjected under a unitary transformation $U$, i.e.,
\begin{equation}
    \hat{\sigma}_{0}^{z}+i\gamma\hat{\sigma}_{0}^{y}=U\left(\hat{\sigma}_{0}^{x}+i\gamma\hat{\sigma}_{0}^{z}\right)U^{\dagger},
\end{equation}
where $U = \frac{1}{\sqrt{2}} \begin{pmatrix} 1 & -i \\ 1 & i \end{pmatrix}$. That means that the \pt \ operator must transform accordingly. The \pt \ operator of $\hat{\sigma}_{0}^{x}+i\gamma\hat{\sigma}_{0}^{z}$ is 
\begin{equation}
    \mathcal{PT}=\hat{\sigma}_{0}^{x}\mathcal{K}.
\end{equation}
Applying the unitary transformation to the \pt \ operator, 
\begin{align}
    U(\mathcal{PT})U^{\dagger}&=\frac{1}{\sqrt{2}} \begin{pmatrix} 1 & -i \\ 1 & i \end{pmatrix}\left(\hat{\sigma}_{0}^{x}\mathcal{K}\right)\frac{1}{\sqrt{2}} \begin{pmatrix} 1 & 1 \\ i & -i \end{pmatrix} \nonumber \\
    &=\frac{1}{\sqrt{2}} \begin{pmatrix} 1 & -i \\ 1 & i \end{pmatrix}\hat{\sigma}_{0}^{x}\frac{1}{\sqrt{2}} \begin{pmatrix} 1 & 1 \\ -i & i \end{pmatrix}\mathcal{K} \nonumber \\
    &=\frac{1}{2}\begin{pmatrix} -2i & 0 \\ 0 & 2i \end{pmatrix}\mathcal{K} \nonumber \\
    &=-i\hat{\sigma}_{0}^{z}\mathcal{K}.
\end{align}
To confirm whether this \pt-operator is indeed a symmetry of $V_{\mathcal{S}}$, we check their commutation, 
\begin{align*}
    [V_{\mathcal{S}},-i\hat{\sigma}_{0}^{z}\mathcal{K}]&=-i[\hat{\sigma}_{0}^{z},\hat{\sigma}_{0}^{z}\mathcal{K}]+\gamma[\hat{\sigma}_{0}^{y},\hat{\sigma}_{0}^{z}\mathcal{K}] \nonumber \\
    &=-i\left(\mathcal{K}-\hat{\sigma}_{0}^{z}\mathcal{K}\hat{\sigma}_{0}^{z}\right)+\gamma\left(\hat{\sigma}_{0}^{y}\hat{\sigma}_{0}^{z}\mathcal{K}-\hat{\sigma}_{0}^{z}\mathcal{K}\hat{\sigma}_{0}^{y}\right) \nonumber \\
    &=-i\left(\mathcal{K}-\mathcal{K}\right)+\gamma\left(\hat{\sigma}_{0}^{y}\hat{\sigma}_{0}^{z}\mathcal{K}+\hat{\sigma}_{0}^{z}\hat{\sigma}_{0}^{y}\mathcal{K}\right) \nonumber \\
    &=\gamma\{\hat{\sigma}_{0}^{y},\hat{\sigma}_{0}^{z}\}\mathcal{K}=0.
\end{align*}
Then, the appropriate \pt \ operator for $V_{\mathcal{S}}$ is $-i\hat{\sigma}_{0}^{z}\mathcal{K}$.

\section{Review on Decoherence and Einselection in the Central Spin Model}\label{Review_Cucchietti}
\subsection{Decoherence in the Central Spin Model}
This appendix briefly reviews the decoherence formalism developed by Cucchietti \textit{et al.}~\cite{Cucchietti}, which serves as the basis for the pseudo-Hermitian extension presented in the main text.

The central spin model is one of the most widely used theoretical models for studying decoherence in quantum systems, particularly in spin-based quantum information processing~\cite{Prokofev2000, Schliemann2003, Zurek2003, Cucchietti}. In this model, the interaction between a central spin and its surrounding spin bath is described by the interaction Hamiltonian
\begin{equation}
    H_{\mathcal{SE}}=\sum_{k=1}^{N}g_{k}\hat{\sigma}_{0}^{i}\otimes\hat{\sigma}_{k}^{i}.
\end{equation}
Here, $g_{k}$ denotes the interaction strength between the central spin and the $k$th bath spin, while $\hat{\sigma}_{k}^{i}$ is the Pauli operator acting along the $i$th direction. Following the formulation of Ref.~\cite{Cucchietti}, we consider the initial product state
\begin{equation}
    \ket{\Psi\!\left(0\right)}=\left(a\ket{0}+b\ket{1}\right)\left(\sum_{n=0}^{2^{N}-1}c_{n}\ket{n}\right),
\end{equation}
where $\ket{0}$ and $\ket{1}$ are eigenstates of the interaction Hamiltonian acting on the central spin with eigenvalues $+1$ and $-1$, respectively. The coefficients $a$, $b$, and $c_{n}$ are complex probability amplitudes satisfying the normalization conditions
\begin{equation}
    |a|^{2}+|b|^{2}=1,
    \qquad
    \sum_{n}|c_{n}|^{2}=1,
\end{equation}
and $\ket{n}$ denotes the $n$th basis state of the spin bath. The state evolves according to
\begin{equation*}
    \ket{\Psi\!\left(t\right)}=a\ket{0}\left(\sum_{n=0}^{2^{N}-1}c_{n}e^{-iB_{n}t}\ket{n}\right)+b\ket{1}\left(\sum_{n=0}^{2^{N}-1}c_{n}e^{iB_{n}t}\ket{n}\right),
\end{equation*}
where
\begin{equation}
    B_{n}=\sum_{k=1}^{N}\left(-1\right)^{n_{k}}g_{k},
\end{equation}
and $n_{k}$ denotes the $k$th binary digit of the integer $n$. The corresponding density operator is
\begin{equation}
    \hat{\rho}\!\left(t\right)=\ket{\Psi\!\left(t\right)}\bra{\Psi\!\left(t\right)}.
\end{equation}
Tracing over the degrees of freedom of the spin bath yields the reduced density matrix of the central spin~\cite{Zurek2003,Cucchietti},
\begin{align}
    \hat{\rho}_{\mathcal{S}}\!\left(t\right)&=|a|^{2}\ket{0}\bra{0}+|b|^{2}\ket{1}\bra{1} \nonumber \\
    &+ab^{*}r\!\left(t\right)\ket{0}\bra{1}+a^{*}br^{*}\!\left(t\right)\ket{1}\bra{0},
\end{align}
where $r\!\left(t\right)$ is the decoherence factor, also known as the decoherence function or Loschmidt amplitude in the central-spin literature~\cite{Cucchietti,Gorin2006}. The decoherence factor determines the decay of the off-diagonal coherence of the reduced density matrix. This can be expressed as the Fourier transform of the characteristic function,
\begin{equation}
    r\!\left(t\right)=\int e^{-iBt}\Phi\!\left(B\right)\,dB,
    \label{decoherence_factor}
\end{equation}
where $\Phi\!\left(B\right)$ denotes the probability distribution of the effective coupling energies~\cite{Cucchietti,Gorin2006}. For sufficiently large bath, the discrete distribution of coupling energies may be approximated by a continuous distribution. Here, it is assumed that the effective coupling energies are statistically independent and possess a finite variance. The central limit theorem implies that the characteristic function approaches a Gaussian distribution~\cite{Cucchietti,Feller1968},
\begin{equation}
    \Phi\!\left(B\right)=\frac{1}{\sqrt{2\pi s_{N}^{2}}}\exp\!\left[-\frac{\left(B-\bar{B}_{N}\right)^{2}}{2s_{N}^{2}}\right].
\end{equation}
Here, $\bar{B}_{N}$ and $s_{N}^{2}$ denote the mean and variance of the coupling-energy distribution, respectively. Taking its Fourier transform, we get the well-known Gaussian decoherence in the central spin model,
\begin{equation}
    r(t)=e^{-s_{N}^{2}t^{2}/2}.
\end{equation}
The decoherence timescale $\tau_{D}$ in this model is 
\begin{equation}
    \tau_{D}=\sqrt{\frac{2}{s_{N}^{2}}}.
\end{equation}

\subsection{Einselection in the Central Spin Model}
Rather than analyzing the reduced density matrix directly, it is often convenient to employ the equivalent Bloch-vector representation~\cite{Nielsen2010},
\begin{equation}
    \hat{\rho}=\frac{1}{2}\left(\mathbb{I}+\vec{p}\cdot\vec{\hat{\sigma}}\right),
    \label{Bloch_vector_decomposition}
\end{equation}
where $\mathbf{p}$ is the Bloch vector of the central spin. This enables a geometric visualization of the time evolution as a contraction of the Bloch sphere. The pointer states are the points on the sphere that remain farthest from the origin at time $t$. The time evolution of the Bloch vector is given by
\begin{equation}
    \vec{p}\!\left(t\right)=\int \hat{U}_{B}\!\left(t\right)\vec{p}\!\left(0\right)\hat{U}_{B}^{\dagger}\!\left(t\right)\Phi\!\left(B\right)\,dB,
\end{equation}
where $U_{B}\!\left(t\right)$ denotes the time-evolution operator corresponding to a particular effective coupling energy $B$. The long-time behavior of the Bloch vector determines the preferred pointer states of the system, namely the stable states selected through the interaction with the environment through environment-induced superselection (einselection)~\cite{Zurek_1981, Zurek_1982, Zurek2003}. These pointer states, which are selected after a very long time, are typically time-independent and arise from the competing internal dynamics and interactions. For a given central spin (spin-$1/2$) immersed in a static bath with Hamiltonian 
\begin{equation}
    \hat{H}=\hat{H}_{\mathcal{S}}+\hat{H}_{\mathcal{SE}},
\end{equation}
such that $[\hat{H},\hat{H}_{\mathcal{S}}]\neq0$ and $\hat{H}_{\mathcal{SE}}=\hat{V}_{\mathcal{S}}\otimes{\hat{V}_{\mathcal{E}}}$, the time evolution operator
\begin{equation}
    U(t)=\mathrm{exp}\left[-i\left(\hat{H}_{\mathcal{S}}+\hat{H}_{\mathcal{SE}}\right)t\right]
\end{equation}
can be expanded into
\begin{equation}
    U(t)=\sum_{n=0}^{2^{N}-1} \hat{U}_{B_{n}}(t)\otimes \ket{n}\bra{n}.
\end{equation}
Here, $\{\ket{n}\}$ and $B_{n}$ are the set of the eigenstates and the corresponding eigenvalues of $\hat{V}_{\mathcal{E}}$. This means that we can represent $\hat{V}_{\mathcal{E}}$ as
\begin{equation}
    \hat{V}_{\mathcal{E}}=\sum_{n=0}^{2^{N}-1}B_{n}\ketbra{n}{n}.
\end{equation}
Then, 
\begin{equation}
    \hat{U}_{B_{n}}(t)=\mathrm{exp}\left[-i\left(\hat{H}_{\mathcal{S}}+B_{n}\hat{V}_{\mathcal{S}}\right)t\right].
    \label{rotation}
\end{equation}
Equation (\ref{rotation}) can be interpreted as a rotation for a given Bloch vector along the plane that $\hat{H}_{\mathcal{S}}$ and $\hat{V}_{\mathcal{S}}$ act on. Expanding equation (\ref{rotation}) into a power series and exploiting the algebra of Pauli matrices, we can get
\begin{equation}
    \hat{U}_{B_{n}}(t)=I\cos{\left(\Omega_{n}t\right)}-i\frac{\left(B_{n}\hat{V}_{\mathcal{S}}+\hat{H}_{\mathcal{S}}\right)}{\Omega_{n}}\sin{\left(\Omega_{n}t\right)},
\end{equation}
where $\Omega_{n}^{2}=\frac{1}{2}\mathrm{Tr}\left[\left(B_{n}\hat{V}_{\mathcal{S}}+\hat{H}_{\mathcal{S}}\right)^{2}\right]$. 


\bibliography{References}

\end{document}